\documentclass[10pt]{amsart}
\usepackage{graphicx}
\usepackage{amssymb}
\usepackage{mathrsfs}
\usepackage{epstopdf}
\usepackage{hyperref}

\theoremstyle{definition} 
\theoremstyle{plain}      
\theoremstyle{plain}      
\theoremstyle{plain}      
\theoremstyle{plain}      
\theoremstyle{plain}      \newtheorem{proposition}{Proposition}
\theoremstyle{plain}      \newtheorem{theorem}{Theorem}
\theoremstyle{remark}     \newtheorem{remark}{Remark}
\theoremstyle{remark}     
\theoremstyle{plain}

\title[Constant-length random substitutions]{Constant-length random 
substitutions and Gibbs measures}
\author{C. Maldonado, L. Trejo-Valencia and E. Ugalde}
\address{Instituto de F\'isica, Universidad Aut\'onoma de San Luis Potos\'i, 
Avenida Manuel Nava 6, Zona Universitaria, 78290 San Luis Potos\'\i , M\'exico.}
\email{liliana@ifisica.uaslp.mx, ugalde@ifisica.uaslp.mx}
\address{Divisi\'on de Matem\'aticas Aplicadas, Instituto Potosino de Investigaci\'on
Cient\'ifica y Tecnol\'ogica, Camino a la Presa San Jos\'e 2055, Lomas 4 secci\'on, 
78216 San Luis Potos\'\i , M\'exico.} 
\email{cesar.maldonado@ipicyt.edu.mx}

\date{}                                         

\begin{document}

\begin{abstract}
This work is devoted to the study of processes generated by random 
substitutions over a finite alphabet. We prove, under mild conditions 
on the substitution's rule, the existence of a unique process which remains 
invariant under the substitution, and exhibiting polynomial decay of correlations. 
For constant-length substitutions, we go further by proving  
that the invariant state is precisely a Gibbs measure which can be obtained
as the projective limit of its natural Markovian approximations. We close the paper 
with a class of substitutions whose invariant state is the unique Gibbs measure
for a hierarchical two-body interaction.
\end{abstract}

\maketitle

\bigskip

\section{Introduction.}

\subsection{} The systems we will consider were introduced to describe the 
evolution of genome sequences, and in general are aimed at explain the pattern 
correlation which can be observed in real genome sequences. 
One of the first examples of this was proposed by 
W. Li~\cite{1989Li} as a simple model exhibiting some spatial 
scaling properties. It was subsequently used to understand the scaling properties 
and the long-range correlations found in real DNA 
sequences~\cite{1995BuldyrevEtAl,1991Li,1992LiKaneko,1997Li,1992PengEtAl}. 
In~\cite{1989GodrecheLuck}, Godr\`eche and J. M. Luck studied similar 
systems used to generate 
quasi-periodic structures by means of random inflation rules.  
From the mathematical point of view, Li's model belongs to the class of random 
substitution, which has attracted some attention in recent years 
(see~\cite{2012Koslicki} and references therein), but whose origin can be traced back
at least until Peyri\`ere's paper~\cite{1981Peyriere}. Peyri\`ere's and Koslicki's 
random substitutions constitute a class of 
Markov processes on a countable set consisting in finite strings over a fixed finite 
alphabet.
The random substitution acts by replacing letters by words, according to a certain 
stochastic rule, and the mathematical questions concern the asymptotic behavior of the
system. Our notion of random substitution
is equivalent to Peyri\`ere's and Koslicki's, with the only difference that we place
ourselves from the very beginning, in the framework of infinite sequences. Hence, 
instead of a Markov chain over a countable set, we deal with a Markov process over the 
set of 
infinite sequences. This is the approach used by Malychev~\cite{1998Malyshev}
and Rocha and coauthors~\cite{2011Rocha&al}. 
In a previous work~\cite{2013SalgadoUgalde}, we studied Li's model 
from this point of view, and we proved the existence of an invariant state exhibiting 
polynomial decay of correlation.

\medskip \noindent \subsection{}  
Loosely speaking, a random substitution can be described as follows. We are 
given a finite set of symbols $A$,  or alphabet, and a collection 
$S:=\{\sigma_1,\sigma_2\ldots,\sigma_m\}$ 
of substitutions,  {\it i.e.} of functions 
$\sigma:A\to A^+:=\cup_{n\in{\mathbb N}} A^n$ 
replacing symbols of the alphabet by finite strings. We extend this action 
coordinate-wise to the set $A^{\mathbb N}$ of infinite strings from $A$. 
Hence, from randomly chosen sequences 
$x_1x_2\cdots x_n\cdots \in A^{\mathbb N}$ and $s_1s_2\cdots s_n\cdots \in 
S^{\mathbb N}$, 
we obtain the sequence $s_1(x_1)s_2(x_2)\cdots s_n(x_n)\cdots $ by concatenation 
of the words $s_i(x_i)\in A^+$ with $i\in {\mathbb N}$. 
By iterating this procedure, we obtain a random sequence of strings in 
$A^{\mathbb N}$, which is supposed to converge in a yet to specify probabilistic 
sense. 

\medskip \noindent \subsection{} 
Our goal in this paper is threefold. First, we aim to establish general and easily
verifiable conditions on the random substitution, ensuring the existence and uniqueness
of an invariant state. We pretend also to establish general conditions for the invariant 
state to exhibit the polynomial decay of correlations. Finally, we want to characterize the 
invariant state from the point of view of the thermodynamic formalism, furnishing a 
description of it in terms of an interaction potential. 
To these aims, we organized the paper as follows: 
In Section~\ref{sec:generalities} we set up 
the mathematical framework, fix the notations, and prove some basic general results 
concerning random substitutions. Section~\ref{sec:constant-length} is devoted to our 
main 
results, which concern
the constant-length case. There we prove that the unique invariant state is a Gibbs 
measure, 
which is the projective limit of a sequence of Markovian approximations.
In Section~\ref{sec:two-bodyinteractions} we examine a class of random substitutions 
whose invariant
state is the unique Gibbs measure for a hierarchical two-body interaction potential. 
We finish with some closing remarks and comments.

\bigskip

\section{Generalities}\label{sec:generalities}  
\subsection{}
Let $A$ be a finite set, with the discrete topology and let us supply 
$A^{\mathbb N}$ with the corresponding product topology and Borel sigma-algebra. 
We will consider the convex set ${\mathcal M}(A^{\mathbb N})$ of all 
Borel probability measures and its proper subset ${\mathcal M}^+(A^{\mathbb N})$
containing all the probability measures with full support, {\it i. e.}, $\mu(B)>0$ 
for each open ball $B\subset A^{\mathbb N}$.  
To each $n\in {\mathbb N}$ and every finite strings $a:=a_1a_2\cdots a_{N}\in A^+$ 
we associate the cylinder set, 
$[a]_n=\{x\in A^{\mathbb N}:\ x_n x_{n+1}\ldots x_{n+N-1}=a\}$. To simplify the 
notations, we will use $[a]$ instead of $[a]_1$ unless it is necessary to specify. 
For each infinite string $x\in A^{\mathbb N}$ and integers $1\leq \ell\leq  m$, 
we will denote the finite substring $x_\ell x_{\ell+1}\cdots x_{m}$ by 
$x_\ell^m$. 
Similarly, we will denote by $\mu_\ell^m$ the marginal of the measure 
$\mu\in{\mathcal M}\left(A^{\mathbb N}\right)$ corresponding to the coordinates
$\ell,\ell+1,\ldots,n$.  

\medskip\noindent
We will consider two different metrics on ${\mathcal M}(A^{\mathbb N})$, one 
compatible with the vague topology and another one, generating a finer topology. 
The vague topology is generated by the distance 
\[
D(\mu,\mu'):=\sum_{N\in{\mathbb N}} 2^{-N}\sum_{a\in A^N}|\mu[a]-\mu'[a]|.
\]
A finer topology is obtained from the {\bf\em projective distance}
$\rho: {\mathcal M}^+(A^{\mathbb N})\times{\mathcal M}^+(A^{\mathbb N})\to [0,1]$ 
given by
\[
\rho(\mu,\mu')=
\sup_{N\in{\mathbb N}}\max_{a\in A^N}
           \frac{1}{N}\left|\log\frac{\mu[a]}{\mu'[a]}\right|.
\]

\medskip\noindent 
\begin{remark}
In~\cite{2015TrejoUgalde} we studied the salient features of the topology generated 
by the projective distance. There it was proved that ${\mathcal M}(A^{\mathbb N})$
is complete and non-separable with respect to $\rho$, so that the topology generated
by $\rho$ is strictly finer than the vague topology.
\end{remark}

\medskip\noindent 
\subsection{} 
A substitution is any map $ \sigma: A\to A^+$ replacing a symbol by a finite string. 
A string of substitutions $s_1s_2\ldots s_N$ defines a map $s:A^N\to A^+$ by 
concatenation of the images of each individual substitution, {\it i. e.}, 
$s(a)=s_1(a_1)s_2(a_2)\cdots s_N(a_N)$, for each $a\in A^N$. 

\medskip \noindent 
The {\bf\em minimal length} of the substitution $S$ is the integer
$\ell_S:=
\min\{\ell\in {\mathbb N}:\, \cup_{\sigma\in S} \sigma(A)\cap A^\ell\neq \emptyset\}$.  
We similarly define $L_S$ to be the {\bf\em maximal length} of the substitution $S$. 
In the case $\ell_S=L_S$ we have a {\bf\em constant-length} substitution. 

\medskip \noindent Given a finite collection 
$S$ of substitutions, we consider the product space $S^{\mathbb N}$ of infinite 
strings of substitutions supplied with the product topology and corresponding 
sigma-algebra. Consider $\nu\in {\mathcal M}^+(S^{\mathbb N})$, in the set 
of fully supported probability measures, and define the transformation 
${\mathbb S}_{\nu}:
{\mathcal M}^+(A^{\mathbb N})\to{\mathcal M}^+(A^{\mathbb N})$
such that
\begin{equation}\label{eq:substitution}
{\mathbb S}_{\nu}\mu[a]=\sum_{s(b)\sqsupseteq a} \nu[s]\mu[b],
\end{equation} 
for each $a\in A^N$. The sum in the right-hand side of the equation runs over all 
strings $s\in S^N$ and $b\in A^N$ such that $s(b)_1^N=a$, which we denote by 
$s(b)\sqsupseteq a$. The transformation ${\mathbb S}_{\nu}$ is the {\bf\em random 
substitution} defined by $\nu$. 

\medskip \noindent An {\bf\em invariant state} for the substitution 
${\mathbb S}_\nu$ is any measure $\mu\in {\mathcal M}^+(A^{\mathbb N})$ such that 
${\mathbb S}_{\nu}(\mu)=\mu$. 
Schauder-Tychonoff Theorem ensures the existence of invariant states. We will be 
interested in the contractive case, for which there exists a unique invariant state 
$\mu_\nu\in {\mathcal M}(A^{\mathbb N})$ such that 
$\mu_\nu=\lim_{n\to \infty} {\mathbb S}_\nu^{\circ n}\mu$ for each 
$\mu\in {\mathcal M}(A^{\mathbb N})$.
We will consider convergence in both $D$ and $\rho$ distances. 

\medskip\noindent
\subsection{}
We will say that the finite collection $S$ of substitutions on $A$ is 
{\bf\em primitive} if for each $N\in {\mathbb N}$ there exists 
$n_N\in{\mathbb N}$ such that for each $n\geq n_N$ and $a,b\in A^N$ there
exists a sequence of substitution strings $s^{(1)},s^{(2)},\ldots s^{(n)}\in S^N$
such that $s^{(n)}\circ s^{(n-1)}\circ\cdots\circ s^{(1)}(b)\sqsupseteq a$. 
Primitive substitutions are well behaved in the sense that when $S$ is primitive, 
random substitutions defined by a fully supported measure 
$\nu\in {\mathcal M}^+(S^{\mathbb N})$ have a unique 
${\mathbb S}_{\nu}$-invariant state, i. e., a unique probability measure 
$\mu_\nu\in {\mathcal M}^+(A^{\mathbb N})$ such that 
${\mathbb S}_{\nu}(\mu_\nu)=\mu_\nu$. To be more precise, we have the following.

\medskip\noindent
\begin{theorem}[{\bf Convergence}]\label{teo:vagueconvergence}
Let $S=\{\sigma_i:\ A\to A^+,\ i=1,\ldots,m\}$ be primitive. Then, for each 
fully supported measure $\nu\in {\mathcal M}^+(S^{\mathbb N})$, there exists a 
unique invariant state $\mu_\nu\in {\mathcal M}^+(A^{\mathbb N})$, 
such that for each $\mu\in {\mathcal M}(A^{\mathbb N})$, 
$\lim_{n\to\infty} {\mathbb S}_\nu^{\circ n}\mu=\mu_\nu$ in the vague topology.
\end{theorem}

\medskip\noindent
\begin{remark}
It is worth noticing Sinai's work on self-similar distributions~\cite{{1975Sinai}}, 
which can be related to our substitution-invariant states. 
Self-similar distributions are the invariant states to the action of 
a semigroup of transformations involving averaging over block and 
normalization of the lattice over those blocks, which in some sense is 
the oposite of a random substitution.
\end{remark}

\medskip\noindent
\begin{proof}
The result follows, almost straightforwardly, from Perron-Frobenius theorem for 
primitive matrices and Kolmogorov's consistency theorem. 

\noindent Fix $N\in {\mathbb N}$ and let $M_N:A^N\times A^N\to [0,1]$ denote 
the probability transition matrix 
\[M_N(a,b)=\sum_{s(b)\sqsupseteq a} \nu[s].\]
According to Eq.~\eqref{eq:substitution}, for each 
$\mu\in{\mathcal M}(A^{\mathbb N})$ we have 
$({\mathbb S}_{\nu}\mu)_1^N = M_N \mu_1^N$, where $\mu_1^N$ denotes the marginal of 
$\mu\in {\mathcal M}(A^{\mathbb N})$ on the first $N$ coordinates. 
Since the collection $S$ is
primitive and $\nu$ is fully supported, then the matrix $M_N$ is primitive as well. 
Therefore, Perron-Frobenius applies, ensuring the existence of a unique probability 
vector $v_N\in (0,1)^{A^N}$ such that
\[
\lim_{N\to\infty}({\mathbb S}_\nu^{\circ n}\mu)_1^N\equiv
\lim_{n\to\infty}(M_N)^n\mu_1^N=v_N,
\]
for each $\mu\in{\mathcal M}(A^{\mathbb N})$.  

\noindent Fix $n\in {\mathbb N}$ and let $\mu^{(n)}:={\mathbb S}_\nu^{\circ n}\mu$, 
then 
\begin{eqnarray*}
\sum_{a_{N+1}}\left(\mu^{(n+1)}\right)_1^{N+1}(a)&=& 
                \sum_{a_{N+1}}\sum_{s(b)\sqsupseteq a} \nu[s]\mu^{(n)}[b]\\
                                  &=& 
                \sum_{s_1^N(b_1^N)\sqsupseteq a_1^N} \nu[s_1^N]\mu^{(n)}[b_1^N]
                                   = \left(\mu^{(n+1)}\right)_1^{N}(a_1^N),
\end{eqnarray*}
for each $a\in A^{N+1}$. From here it follows, by taking the limit $n\to\infty$, 
that $\{v_N: (0,1)^{A^N}\to (0,1)\}$ is a compatible family of probability vectors, 
i.~e., $\sum_{a_{N+1}} v_{N+1}(a)=v_N(a_1^N)$, for each $N\in {\mathbb N}$ and 
$a\in A^{N+1}$. Hence, by virtue of Kolmogorov's consistency theorem, there 
exists a unique, well defined measure $\mu_\nu\in{\mathcal M}^+(A^{\mathbb N})$ 
such that 
\[
\lim_{n\to\infty}({\mathbb S}_\nu^{\circ n}\mu)_1^N=(\mu_\nu)_1^N\equiv v_N,
\]
which gives the desired result. 
\end{proof}

\medskip\noindent\begin{remark}
Primitivity in deterministic substitutions, which is behind the minimality of 
the associated dynamical system (see for instance~\cite{2010QueffelecBook}), 
can be readily checked. We have a simple characterization of primitivity 
for constant-length substitutions, and although we do not have such a complete 
characterization for general random substitutions, nevertheless, we can  
establish easily verifiable sufficient conditions 
(see Appendix~\ref{primitivity}).
\end{remark}

\medskip\noindent
\subsection{}
In the case of primitive constant-length substitutions with length strictly larger 
than one, we can prove that the unique ${\mathbb S}_\nu$-invariant state has 
polynomial decay of correlations. Indeed, we have the following.

\begin{theorem}[{\bf Decay of correlations}]\label{teo:decayofcorrelation} 
Let $S$ be primitive constant-length substitution of length $L>1$ and let 
$\nu\in {\mathcal M}(S^{\mathbb N})$ be a fully supported product measure. 
Then there are constants $C,\gamma >0$ and $n_0\in {\mathbb N}$ such that for all 
$a,b\in A$, and $n\geq n_0$ we have
\[
1-C\,n^{-\gamma}\leq 
\frac{\mu_\nu([a]_1\cap [b]_n)}{\mu_\nu[a]_1\,\mu_\nu[b]_n}
                \leq 1+C\,n^{-\gamma}.
\]
\end{theorem}

\medskip\noindent The proof is again based on the classical Perron-Frobenius 
Theorem for finite matrices as stated in~\cite{2003Cavazos}. 

\medskip \noindent
\begin{proof} For each $n\in{\mathbb N}$ and $1\leq j\leq L$ 
let $M_{n,j}:A\times A\to [0,1]$ be such that 
$M_{n,j}(a,c)=\sum_{\sigma(c)_j=a}\nu[\sigma]_n$, which is just
the probability of obtaining a string with the symbol $a$ at the $(n-1)L+j$-th 
position by substitution of the symbol $c$. It is not difficult to verify 
that $M_{1,1}$ inherits the primitivity of $S$, hence,  
by virtue of Birkhoff's version of Perron-Frobenius theorem (as stated 
in~\cite{2003Cavazos}), and taking into account that $\mu_\nu$ is the unique 
${\mathbb S}_\nu$-invariant state, we necessarily have 
\[
\mu_\nu[a]_1\,e^{-C_v\eta^k}\leq M_{1,1}^k (a,c)\leq \mu_\nu[a]_1\,e^{C_v\eta^k}.
\]
for every $k\in {\mathbb N}$. Here $C_v > 0$ and $\eta=\tau^{1/\ell}\in [0,1)$, 
with $\tau\in[0,1)$ and $\ell\in{\mathbb N}$ are respectively the Birkhoff's coefficient 
and the primitivity index of the matrix $M_1$.

\medskip \noindent Now, by using the base-$L$ expansion of the integers 
we can link the different marginals of $\mu_\nu$ to 
its first marginal. Indeed, let $n-1=\sum_{k=0}^q e_k L^k$ and for each 
$0\leq k\leq q$ define $n_k-1=\sum_{i=0}^k L^i e_{k-i}$, then
\[
\mu_\nu[b]_n= \sum_{c\in A}\left(\prod_{k=0}^q 
M_{n_k,e_k+1}\right)(b,c)\,\mu_\nu[c]_1.
\]
The integers $n_0,n_1,\ldots,n_q=n$ are the successive positions in a 
sequence of substitutions, connecting the first position to the $n$-th position. 
Since $\nu$ is a product measure, then
\[
\mu_\nu([a]_1\cap [b]_n)=\sum_{c\in A} M_{1,1}^q(a,c)\, 
              \left(\prod_{k=0}^q M_{n_k,e_k+1}\right)(b,c)\,\mu_\nu[c]_1,
\]
for each $n\in{\mathbb N}$ and $a,b\in A$. Now, taking into account that $M_{1,1}$ is 
primitive, we finally obtain
\[
e^{-C_v\eta^q}\mu_\nu[a]_1\,\mu_\nu[b]_n \leq \mu_\nu([a]_1\cap [b]_n)\leq 
                 e^{C_v\eta^q}\mu_\nu[a]_1\,\mu_\nu[b]_n
\]
and the result follows by taking $n_0=L^{q_0}$ with $C_v\eta^{q_0}< 1/2$, 
$C=2 C_v$ and $\gamma=|\log\eta/\log L|$.
\end{proof}

\medskip\noindent\begin{remark}
All of the computations carried out on the previous proof can be adapted to non-constant 
length substitutions provided the minimal length $\ell_S > 1$. 
In that case we have a polynomial upper bound for the decay of correlation, 
with an exponent $-\gamma=\log\tau/(\ell\log L_S)$, where 
$\tau$ and $\ell$ are respectively the Birkhoff's coefficient and primitivity index of $M_{1,1}$. 
Polynomial decay of correlations can also appear in non-constant length  
substitutions for which $\ell(S)=1$, as we proved in~\cite{2013SalgadoUgalde}. In the last 
case some additional properties of the substitution play a role in the derivation of the 
polynomial decay of correlations.
\end{remark}

\medskip\noindent
\subsection{}  
Let ${\mathcal L}=\{\Lambda \subset {\mathbb N}:\ \#\Lambda<\infty\}$ and
for each $n\in {\mathbb N}$ and $n\in \Lambda \subset {\mathcal L}$, let 
$\Phi_{n,\Lambda} : A^{\mathbb N}\to {\mathbb R}$ be a measurable function. 
The collection 
$\Phi:=\{\Phi_{n,\Lambda}\}_{n\in\Lambda\subset {\mathcal L}}$ is 
an {\bf\em interaction potential}, if for each $a\in A^{\mathbb N}$ and 
$\Lambda\in{\mathcal L}$, the {\bf\em total energy} of $a$ in $\Lambda$, 
\begin{equation}\label{eq:energy}
H_\Lambda(a):=\sum_{n\in\Lambda}\sum_{\Lambda'\ni n}\Phi_{n,\Lambda'}(a)<\infty.
\end{equation}

\medskip\noindent We will distinguish the particular case of 
{\bf\em two-body interaction potential}, which corresponds to interactions 
$\Phi=\{\Phi_{n,\Lambda}\}_{n\in\Lambda\subset {\mathcal L}}$ satisfying
$\Phi_{n,\Lambda}=0$ whenever $\#\Lambda\neq 2$. 

\medskip\noindent
\begin{remark}
Here we are using a slightly different definition of interaction potential than
the classical one as it appears in~\cite{1988GeorgiiBook} or~\cite{2004RuelleBook}. 
Our adaptation does not affect the development of the theory since the 
Gibbs factors defined by using our definition are the same as those defined by
using the classical defintion.
\end{remark}

\medskip\noindent 
A measure $\mu\in {\mathcal M}^+(A^{\mathbb N})$ is a {\bf\em Gibbs measure} 
for the interaction potential $\Phi$ if for each finite 
$\Lambda\subset {\mathbb N}$ and fixed $a\in A^{\mathbb N}$ we have
\[
\mu\{x_\Lambda = a_\Lambda |x_{\Lambda^c}=a_{\Lambda^c}\}:=\lim_{N\to\infty}
\mu\left([a_{\Lambda}]|\,a_{\{1,2,\ldots,N\}\setminus\Lambda}]\right)
=\frac{e^{-H_\Lambda(a_\Lambda\oplus a_{\Lambda^c})} }{
\sum_{a'_\Lambda\in A^\Lambda}e^{-H_\Lambda(a'_\Lambda\oplus a_{\Lambda^c}) }}.
\] 
Here we use $a_F$ to denote the projection of $a\in A^{\mathbb N}$ on 
the coordinates $F\subset {\mathbb N}$, and 
$a'_\Lambda\oplus a_{\Lambda^c}\in A^{\mathbb N}$ to denote the configuration 
whose projections on $\Lambda$ and $\Lambda^c:={\mathbb N}\setminus\Lambda$ 
coincide with $a'_\Lambda$ and $a_{\Lambda^c}$ respectively. 
We denote by ${\mathcal G}(\Phi)\subset {\mathcal M}(A^{\mathbb N})$ the set of 
all Gibbs measures for the interaction $\Phi$.
Although the functions $\Phi_{n,\Lambda}$ are only required to be Borel measurable 
and satisfying~\eqref{eq:energy}, we will further require them to be 
{\bf\em absolutely summable}, {\it i. e.}, to be such that
\[
\sum_{\Lambda\ni n}||\Phi_{n,\Lambda}||< K_\Phi,
\]
for some constant $K_\Phi>0$ and all $n\in{\mathbb N}$.

\medskip\noindent
\subsection{}
If the interaction potential $\Phi=\{\Phi_{n,\Lambda}\}_{n,\Lambda\in{\mathcal L}}$ 
is absolutely summable, the existence of Gibbs measures for this interaction follows from 
the compactness of ${\mathcal M}(A^{\mathbb N})$ in the vague topology and 
the uniform continuity of the {\bf\em local potentials} 
$\{\phi_n: A^{\mathbb N}\to{\mathbb R}\}_{n\in{\mathbb N}}$ given by 
\[
\phi_n(a):=\sum_{\Lambda\ni n} \Phi_{n,\Lambda}(a).
\]
This is a particular case of a more general result which can be found 
in Georgii's book~\cite{1988GeorgiiBook}. There it is also proved that 
the set ${\mathcal G}(\Phi)$ is a Choquet simplex, which means that each 
$\mu\in {\mathcal G}(\Phi)$ can be decomposed, in a unique way, as a convex
combination of extremal measures. To be more precise, there exists a set 
${\mathcal E}\subset {\mathcal G}(\Phi)$ and for each  $\mu\in {\mathcal G}(\Phi)$ 
there exists a unique Borel probability measure 
$\nu_{\mu}\in {\mathcal M}({\mathcal E})$ such that 
$\mu=\int_{\mathcal E}\eta\ d\nu_{\mu}(\eta)$. It is also proved that different 
extremal measures are mutually singular, \emph{i.~e.} if $\mu,\mu'$ are different 
extremal measures, then $\mu\perp\mu'$.

\medskip\noindent
\begin{remark}
Notice that our notion of absolute summability for the interaction differs from, but it is 
closely related to the classical notion of regularity of local potential one can find 
in~\cite{1998KellerBook}. Notice as well that the total energy of a 
configuration $a\in A^{\mathbb N}$ inside the volume $\Lambda\in{\mathcal L}$, can be 
computed by using the local potentials as $H_\Lambda(a)=\sum_{n\in\Lambda}\phi_n(a)$.
\end{remark}

\medskip\noindent  
\subsection{}
The most generally applicable criterion for uniqueness is Dobrushin's 
condition~\cite{1968Dobrushin}, which depends on the behavior of a correlations 
matrix computed from the interaction potential. Fix $n,\, m\in {\mathbb Z}$ and 
define
\[
C_\Phi(n,m)=\sup_{x,b,c\in A^{\mathbb N}}
\frac{1}{2}\sum_{a_n\in A}
\left|\frac{e^{-H_\Lambda(a_n\oplus b_m\oplus x_{\{m,n\}^c})}}{
     \sum_{a'_n\in A}e^{-H_\Lambda(a'_n\oplus b_m\oplus x_{\{m,n\}^c})}}-
      \frac{e^{-H_\Lambda(a_n\oplus c_m\oplus x_{\{m,n\}^c})}}{
    \sum_{a'_n\in A}e^{-H_\Lambda(a'_n\oplus c_m\oplus x_{\{m,n\}^c})}}\right|,
\]
where $x_n\oplus y_m\oplus z_{\{m,n\}^c}$ representes the obvious concatenation. 
Dobrushin result states that if 
\[
\sup_{n\in {\mathbb N}} \sum_{m\in {\mathbb Z}}C_\Phi(n,m) < 1,
\]
then $\#{\mathcal G}(\Phi)=1$. An easily-checked condition, derived from the 
previous one, is due to Simon~\cite{1979Simon} and depends on the decay of the 
oscillation of the interaction potential. The {\bf\em oscillation} of a function 
$\psi:A^{\mathbb N}\to {\mathbb R}$ is given by 
$\sup_{a,b\in A^{\mathbb N}}|\psi(a)-\psi(b)|$ and denoted ${\rm osc}(\psi)$. 
Adapting from Simon, if the collection of local potentials 
$\{\Phi_{n,\Lambda}:A^{\mathbb N}\to {\mathbb R}\}_{n\in{\mathbb N}}$ is such that
\begin{equation}\label{eq:simoncondition}
\sup_{n\in {\mathbb N}} 
            \sum_{\Lambda\ni n}(\#\Lambda-1){\rm osc}(\Phi_{n,\Lambda}) < 2,
\end{equation}
then $\#{\mathcal G}(\Phi)=1$.
  
\medskip \noindent
\begin{remark}
A complete discussion concerning Gibbs measures, its existence and 
uniqueness, can be found in Georgii's book cited above.
\end{remark}

\bigskip
\section{Constant-length substitutions}\label{sec:constant-length}
\subsection{} 
A {\bf\em Markov measure} is nothing but a Gibbs measure 
$\mu\in {\mathcal G}(\Phi)$ for an interaction potential 
$\Phi=\{\Phi_{n,\Lambda}\}_{n\in\Lambda\in{\mathcal L}}$ 
satisfying $\Phi_{n,\Lambda}\equiv 0$ for all 
$\Lambda\subset\in{\mathcal L}$ such that $\max\Lambda-\min\Lambda > r$ 
for some $r\in {\mathbb N}$. 
The minimal integer $r$ satisfying the above condition is know as {\bf\em the 
range of the interaction potential}. In this case, the corresponding local 
potential $\phi_n$ depends only on sites at distance not larger than $r$ from 
site $n$, and we have
\begin{align*}
\mu\{x_n=a_n|\, x_{\mathbb N\setminus \{n\}}=b\}&=
                              \mu\{x_n=a_n|\, x_m=b_m: \, 0<|m-n|\leq r\}\\
                                              &= 
          \frac{e^{-\phi_n(a_n,b_m:\ 0<|m-n|\leq r)}
              }{\sum_{c\in A}e^{-\phi_n(c,b_m:\ 0<|m-n|\leq r)}}.  
\end{align*}
For Markov measures as defined above, condition~\eqref{eq:simoncondition} is 
trivially satisfied, therefore $\#{\mathcal G}(\Phi)=1$.

\medskip\noindent \subsection{}
Given $\mu\in {\mathcal M}(A^{\mathbb N})$, an 
{\bf\em approximation scheme} is a sequence 
$\{\mu^{(\ell)}\in {\mathcal M}(A^{\mathbb N})\}_{\ell\in{\mathbb N}}$ of Markov 
measures, such that $\lim_{\ell\to\infty}\mu^{(\ell)}=\mu$. The convergence could 
be in the vague topology, or in the projective distance, or in any other topology 
we consider on ${\mathcal M}(A^{\mathbb N})$. Concerning the vague topology we 
already have the following.

\medskip\noindent 
\begin{proposition}[{\bf Vague convergence}] 
Let $S$ be a primitive constant-length substitution of length $L > 1$ and let 
$\nu\in {\mathcal M}(S^{\mathbb N})$ be a fully supported product measure. 
Let $\mu_\nu$ the unique ${\mathbb S}_\nu$-invariant state. Then, for any
product measure $\mu$, the sequence
$\left\{\mu^{(\ell)}:={\mathbb S}_\nu^{\circ \ell}\mu\right\}_{\ell\in{\mathbb N}}$ 
is an approximation scheme for $\mu_\nu$, converging in the vague distance. 
\end{proposition} 

\medskip\noindent \begin{proof}
The vague convergence of $\left\{\mu^{(\ell)}\right\}_{\ell\in{\mathbb N}}$ 
towards $\mu_\nu$ follows from Theorem~\ref{teo:vagueconvergence}. 
We only have to check that $\mu^{(\ell)}$ is a Markov measure. For this, note that 
$\mu^{(\ell)}$ is a block-independent process. Indeed, for each 
$\ell,m\in{\mathbb N}$ and $a\in A^{m\,L^\ell}$ we have
$\mu^{(\ell)}\left[a\right]=\prod_{k=0}^{m-1}
\mu^{(\ell)}\left[a_{kL^\ell+1}^{(k+1)\,L^\ell}\right]_{kL^\ell+1}$,
therefore
\[
\mu^{(\ell)}\left\{x_\Lambda=a_\Lambda|\, x_{\Lambda^c}=b_{\Lambda^c}\right\}
=\mu^{(\ell)}\left\{x_\Lambda=a_\Lambda|\, x_{B(\Lambda,n)\setminus\Lambda}
                               =b_{B(\Lambda,n)\setminus\Lambda}\right\},
\]
where $B(\Lambda,\ell)=\bigcup_{n\in\Lambda}B(n,\ell)$ with
$B(n,\ell):=\lfloor(n-1)/L^{\ell}\rfloor\,L^{\ell}+\{1,2,\ldots,L^{\ell}\}$.
It can be verified that $\mu^{(\ell)}$ is determined by interaction 
potential 
$\Phi^{(\ell)}=\{\Phi^{(\ell)}_{n,\Lambda}\}_{n\in\Lambda\in{\mathcal L}}$ 
given by
\[ \Phi^{(\ell)}_{n,\Lambda}(a)=\left\{
   \begin{array}{cl} - L^{-\ell}\log\mu^{(\ell)}\left[a_{B(n,\ell)}\right]
          & \text{ if } \Lambda=B(n,\ell) \text{ for some } n\in{\mathbb N},\\ 
        0 & \text{ otherwise.} \end{array}\right.
\]
Indeed we have
\[ \mu^{(\ell)}\left\{x_\Lambda=a_\Lambda|\, x_{\Lambda^c}=b_{\Lambda^c}\right\}
= \frac{e^{-\sum_{n\in\Lambda}\Phi^{(\ell)}_{n,B(n,\ell)}
\left(a_\Lambda\oplus b_{\Lambda^c}\right)}
      }{\sum_{a'_\Lambda\in A^{\Lambda}}
          e^{-\sum_{n\in\Lambda}\Phi^{(\ell)}_{n,B(n,\ell)}
             \left(a'_\Lambda\oplus b_{\Lambda^c}\right)} }.
\]
We finish the proof by noting that $\Phi^{(\ell)}$ has range is $L^{\ell}$.
\end{proof}

\medskip\noindent \subsection{}
Now, concerning the convergence in projective distance,
we have the following.

\begin{theorem}[{\bf Projective convergence}]\label{teo:projectiveapproximation} 
Let $S$ be a primitive constant-length substitution of length $L > 1$ and let 
$\nu\in {\mathcal M}(S^{\mathbb N})$ be a fully supported product measure. 
Let $\mu_\nu$ be the unique ${\mathbb S}_\nu$-invariant state and 
$\mu$ the unique product measure such that $\mu[a]_n=\mu_\nu[a]_n$ 
for each $n\in {\mathbb N}$ and $a\in A$. 
If $\rho(\mu,\mu_\nu)<\infty$, then the approximation scheme 
$\left\{
\mu^{(\ell)}:={\mathbb S}_\nu^{\circ \ell}\mu\right\}_{\ell\in{\mathbb N}}$
converges in the projective sence. 
\end{theorem}

\medskip\noindent In Appendix~\ref{boundedness} we established sufficient 
conditions on the product measure $\nu$, ensuring that 
$\rho(\mu,\mu_\nu)<\infty$. Those conditions are satisfied, in particular, 
when $\nu$ is a fully supported Bernoulli measure. 

\medskip\noindent From this point on we will use 
$A\lessgtr B e^{\pm \epsilon}$ and $A\lessgtr B\pm \epsilon$ as shorthand notations 
for $e^{-\epsilon}B \leq A \leq e^{\epsilon} B$ and $B-\epsilon\leq A\leq B+\epsilon$ 
respetively.

\medskip\begin{proof} First note that for $N\in{\mathbb N}$ fixed and 
each $a\in A^{\mathbb N}$, we have $\mu^{(\ell)}[a]=M_N^{\ell}\mu_1^N(a)$, 
with $M_N$ as defined in the proof of Theorem~\ref{teo:vagueconvergence}. 
Since $M_N$ is primitive, then, according to Hilbert's version of the 
Perron-Frobenius Theorem, there are constants $n_N\in{\mathbb N}$ and 
$\tau_N\in [0,1)$ such that 
\[
\frac{{\mu^{(\ell)}}[a]}{\mu_\nu[a]}=
\frac{M_N^{\ell}\mu_1^N(a)}{\mu_\nu[a]}\lessgtr\exp\left(\pm
\tau_N^{\ell}\rho(\mu_\nu,\mu)\right),
\]
for all $\ell\geq n_N$. For each $N\in{\mathbb N}$ let 
$\ell_N\geq n_N$ be such that $\tau_N^{\ell_N}\leq 1/N$, then, 
for every $\ell \geq \ell_2$, we define 
$N(\ell):=\max\{N\in{\mathbb N}: \ell\geq \max(n_N,\ell_N)\}$.
Clearly $N(\ell)\rightarrow \infty$ when $\ell\to\infty$.
With this, 
\begin{equation}\label{eq:shortmarginals}
\exp\left(-\frac{\rho(\mu_\nu,\mu)}{N(\ell)}\right)\leq 
\frac{{\mu^{(\ell)}}[a]}{\mu_\nu[a]}\leq 
\exp\left(\frac{\rho(\mu_\nu,\mu)}{N(\ell)}\right),
\end{equation}
for all $a\in \cup_{n=1}^{N(\ell)} A^{n}$.

\medskip\noindent 
From now on we follow Seneta in~\cite[Lemma 3.1]{2006SenetaBook}. 
Fix $n,\ell\in {\mathbb N}$ and consider the probability transition matrix 
$M_{\ell,n}:A^{n\,L^{\ell}}\times A^{n} \to [0,1]$ such that 
\[
M_{\ell,n}(a,b)=
\sum_{s^{(1)}\cdots s^{(\ell)}(b)=a}\prod_{k=1}^\ell \nu\left[s^{(k)}\right],
\]
which is nothing but the probability of obtaining the string 
$a\in A^{n\,L^\ell}$ by a random substitution of a sequence starting 
with $b\in A^{n}$. Since the $S$ is primitive, then all the strings in 
$A^+$ can be obtained by substitution, therefore 
$\sum_{b\in A^n} M_{\ell,n}(a,b) > 0$ for all $a\in A^{n\,L^\ell}$. 
Now, for every couple of positive probability vectors $u,v: A\to (0,1)$, 
and each $a\in A^{n\,L^\ell}$, we have
\[
\frac{\left(M_{\ell,n}u\right)\, (a)}{\left(M_{\ell,n}v\right)\, (a)}
     =\sum_{b\in A^n}\frac{u(b)}{v(b)} 
     \left(\frac{M_{\ell,n}(a,b) v(b)}{\sum_{b'\in A^n}M_{\ell,n}(a,b')}\right)\in 
     \left[\min_{b\in A^N}\frac{u(b)}{v(b)},\max_{b\in A^n}\frac{u(b)}{v(b)}\right],
\]
since $b\mapsto M_{n,N}(a,b)v(b)/(\sum_{b'\in A^L}M_{n,N}(a,b'))$ defines a 
probability vector. From here, taking $v=(\mu_\nu)_1^n$ and $u=\mu_1^n$ and
reducing to the corresponding marginal, we obtain  
\begin{equation}\label{eq:longmarginals}
\exp\left(-n\rho(\mu_\nu,\mu)\right)\leq \frac{{\mu^{(\ell)}}[a]}{\mu_\nu[a]}
\leq \exp\left(n\rho(\mu_\nu,\mu)\right),
\end{equation}
for each $(n-1)\,L^{\ell} < N\leq n\,L^{\ell}$ and each $a\in A^N$.

\medskip\noindent
From inequalities~\eqref{eq:shortmarginals} and~\eqref{eq:longmarginals}, 
it follows that
\[
\rho(\mu^{(\ell)},\mu_\nu)\leq \rho(\mu,\mu_\nu)\times
\left(\frac{1}{N(\ell)}+\frac{1}{L^{\ell}}\right),
\]
and the result follows.
\end{proof}

\medskip\noindent\begin{remark} It can be shown that in general 
$N(\ell)={\mathcal O}(1/\log(\ell))$. The projective convergence is therefore 
extremely slow. We built the approximation scheme by starting with a product 
measure having the exactly the same one-marginals as the invariant state, but
the result can be extended to schemes starting with any product measure at 
finite projective distance from $\mu_\nu$. 
\end{remark}

\medskip\noindent
\subsection{}
Primitive constant-length substitutions have a nice description
in terms of interaction potentials. We have the following result.

\begin{theorem}[{\bf Gibbsianness of the invariant state}]\label{teo:gibbsianness} 
Let $S$ be a primitive constant-length substitution of length $L > 1$ and let 
$\nu\in {\mathcal M}(S^{\mathbb N})$ be a fully supported product measure. 
Let $\mu_\nu$ the unique ${\mathbb S}_\nu$-invariant state and 
$\mu$ the unique product measure such that $\mu[a]_n=\mu_\nu[a]_n$ 
for each $n\in {\mathbb N}$ and $a\in A$. 
If $\rho(\mu,\mu_\nu)<\infty$, then there exists an interaction
potential $\Phi=\{\phi_{n,\Lambda}\}_{n\in\Lambda\in{\mathcal L}}$
such that ${\mathcal G}(\Phi)=\{\mu_\nu\}$. 
\end{theorem}

\begin{proof} We divide the proof into three parts: First we find an explicit
form for an interaction potential $\Phi$ so that $\mu_\nu\in{\mathcal G}(\Phi)$, 
then we prove that ${\mathcal G}(\Phi)=\{\mu_\nu\}$ under the hypothesis of absolute
summability, which we establish in the final step of the proof.

\medskip\noindent 
\paragraph{\em \underline{Step one: the interaction}}
For each $n,\ell\in{\mathbb N}$ let $B(n,\ell)$ denote the unique
$L$-adic interval of generation $\ell$ containing $n$, {\it i. e.},
$B(n,\ell)=q\,L^\ell+\{1,2,\ldots,L^{\ell}\}$,
where $n=q\,L^\ell+r$, with $0\leq r < L^\ell$.
With this define
\[
\Phi_{n,B(n,\ell)}(a)=\frac{1}{L^\ell}\left(\sum_{k=0}^{L-1}\log\left(
\mu_\nu\left[a_{B(q\,L^{\ell-1}(L+k)+1,\ell-1)}\right]\right)
-\log\left(\mu_\nu\left[a_{B(n,\ell)}\right]\right)\right)
\]
for each $a\in A^{\mathbb N}$. Otherwise, 
if $\Lambda\notin \{ B(n,\ell):\, \ell\in {\mathbb N}\}$, then 
$\Phi_{n,\Lambda}\equiv 0$. A straightforward computation leads to
\[
\mu_\nu\left[a_{\Lambda}\right]=
\exp\left(-\sum_{n\in \Lambda}\sum_{\Lambda'\subset\Lambda}
\Phi_{n,\Lambda'}(a)\right),
\]
for each $\Lambda\in \{ B(n,\ell):\, n, \ell\in {\mathbb N}\}$, if by
convention we fix that
$\Phi_{n,\emptyset}\equiv 0$ for each $n\in {\mathbb N}$.

\medskip\noindent 
\paragraph{\em \underline{Step two: uniqueness}}
Let us assume at the moment, that $\Phi$ is such that
\begin{equation}\label{eq:regularityhyp}
||\Phi_{n,B(n,\ell)}||\leq \frac{K}{L^\ell}, 
\end{equation}
for some $0 \leq K <\infty$ and all $n,\ell\in{\mathbb N}$. 
Under this hypothesis $\Phi$ is absolutely summable. Further more, the local potentials 
$\phi_n(a):=\sum_{\Lambda\ni n} \Phi_{n,\Lambda}(a)$ satisfy
\begin{equation}\label{eq:regularitylocal}
\phi_n(b) \lessgtr \phi_n(a) \pm \frac{2K}{L^\ell(L-1)},
\end{equation}
for all  $\ell\in{\mathbb N}$ and $a,b\in A^{\mathbb N}$ such that 
$a_{B(1,\ell)}=b_{B(1,\ell)}$. 

\medskip\noindent 
From here we follow the ``thermodynamic limit approach'' which 
consists in considering limits of ``finite volume'' versions of Gibbs measures with 
``fixed boundary conditions''. To be more precise, 
fix $a\in A^{\mathbb N}$, and for each $\ell\in{\mathbb N}$ consider the measure
$\mu_{a,\ell}\in {\mathcal M}(A^{\mathbb N})$, with support in the finite set
$X_{a,\ell}:=\{c\in A^{\mathbb N}: c_n=a_n\ \forall n\notin B(1,\ell)\}$, 
given by
\[
\mu_{a,\ell}\{c\}=
\frac{
e^{-\sum_{n\in B(1,\ell)}\phi_n(c)}}{
\sum_{c'\in X_{a,\ell} } e^{-\sum_{n\in B(1,\ell)}\phi_n(c')} 
     }.
\]
We will prove that $\{\mu_{a,\ell}\}_{\ell \in {\mathbb N}}$ converges in 
the vague topology, and that
${\mathcal G}(\Phi)=\{\lim_{\ell\to\infty} \mu_{a,\ell}\}$.
Due to compactness, the sequence 
$\{\mu_{a,\ell}\}_{\ell \in {\mathbb N}}$
has accumulation points. Fix $\mu\in{\mathcal G}(\Phi)$, 
$\Lambda\in{\mathcal L}$ and $\ell\in {\mathbb N}$ such that 
$\max\Lambda \leq L^{\ell}$. From assumption~\eqref{eq:regularitylocal}
we derive,
\begin{align*}
\mu[b_\Lambda]&=\int_{A^{\mathbb N}}
\mu\left\{x_\Lambda=b_\Lambda|x_{{B(1,\ell)}^c}\right\}\ 
d\mu\left(x_{{B(1,\ell)}^c}\right)\\
&=\sum_{c_{B(1,\ell)\setminus\Lambda}}
\int_{A^{\mathbb N}} 
\frac{e^{-\sum_{n\in\Lambda}
\phi_n(b_\Lambda\oplus c_{B(1,\ell)\setminus\Lambda}\oplus x_{B(1,\ell)^c})}
     }{\sum_{c'_{B(1,\ell)}\in A^{B(1,\ell)} } 
e^{-\sum_{n\in\Lambda} 
\phi_n(c'_{B(1,\ell)}\oplus x_{B(1,\ell)^c})} }
d\mu\left(x_{{B(1,\ell)}^c}\right)\\
&\lessgtr e^{\pm 4\, \#\Lambda\, K_\phi\,L^{-\ell}}
\sum_{c_{B(1,\ell)\setminus\Lambda}}
\frac{e^{-\sum_{n\in\Lambda}
\phi_n(b_\Lambda\oplus c_{B(1,\ell)\setminus\Lambda}\oplus a_{B(1,\ell)^c})}
     }{\sum_{c'_{B(1,\ell)}\in A^{B(1,\ell)} } 
e^{-\sum_{n\in\Lambda} 
\phi_n(c'_{B(1,\ell)}\oplus a_{B(1,\ell)^c})} }\\
&\lessgtr e^{\pm 4\, \#\Lambda\,K_\phi\,L^{-\ell}}
\mu_{a,\ell}[b_\Lambda],
\end{align*}
for each $b\in A^{\mathbb N}$. This implies that any accumulation point of
$\{\mu_{a,\ell}\}_{\ell \in {\mathbb N}}$ coincides with $\mu$, and 
this for each $\mu{\mathcal G}(\Phi)$.  
Therefore ${\mathcal G}(\Phi)=\{\lim_{\ell\to\infty} \mu_{a,\ell}\}$ for 
arbitrary $a\in A^{\mathbb N}$.

\medskip\noindent
\paragraph{\em \underline{Third step: absolute summability}} 
To finish the proof, let us establish the validity of 
assumption~\eqref{eq:regularityhyp}. 
For this we use a computation very similar to the one we developed in 
the proof of Theorem~\ref{teo:projectiveapproximation}. 
For $n\in {\mathbb N}$ let $P_{\ell,n}:A^{L^\ell}\times A^L\to [0,1]$
the probability transition matrix such that 
\[
P_{\ell,n}(a,b)=
\sum_{s^{(1)}\cdots s^{(\ell)}(b)=a}\prod_{k=1}^\ell 
\nu\left[s^{(k)}\right]_{[(n-1)/L^{\ell}]L+1},
\]
which give the probability of obtaining $a\in A^{L^\ell}$ at position
$q\,L^\ell+1:=[(n-1)/L^{\ell}]L^\ell+1$, which is the first position in 
$B(n,\ell)$, starting from $b\in A^L$ at position $q\,L+1$. 
Notice that 
\[
e^{-\Phi_{n,B(n,\ell)}(a)}\equiv 
\frac{\mu_\nu\left[a_{B(n,\ell)}\right] 
    }{\prod_{k=0}^{L-1}
\mu_\nu\left[a_{B(q\,L^{\ell-1}(L+k)+1,\ell-1)}\right]
    }=\frac{\left(P_{\ell,n}\, v\right)(a)
          }{\left(P_{\ell,n}\, u\right)(a)}
\]
by taking $u,v:A^{L}\to (0,1)$ the marginals $u=\mu_{qL+1}^{qL+L}$ and
$v=(\mu_\nu)_{qL+1}^{qL+L}$ respectively. Since 
$\sum_{b\in A^L} P_{\ell,n}(a,b)>0$ for each $a\in A^{L^\ell}$, following
the same computations as in the proof of Theorem~\ref{teo:projectiveapproximation},
we obtain 
\[e^{-\Phi_{n,B(n,\ell)}(a)}\lessgtr e^{\pm L\,\rho(\mu,\mu_\nu)},\]
for all $a\in A^{L^\ell}$, therefore
$||\Phi_{n,B(n,\ell)}(a)||\leq K:=L\rho(\mu_\nu,\mu)$, 
and the proof is done.
\end{proof}


\bigskip
\section{Two-body interactions}\label{sec:two-bodyinteractions}

\medskip\noindent
\subsection{}
\noindent Let $A$ be a finite alphabet and for each $a\in A$ let
$\pi_a:A\to A$ be a permutation. With this define a collection of substitutions
$S=\{\sigma_b:A\to A^2:\ b\in A\}$ such that $\sigma_b(a)=(\pi_ab)a$ for
each $a,b\in A$. Now, let $\nu\in {\mathcal M}(S^{\mathbb N})$ be the 
Bernoulli measure such that $\nu[\sigma_b]=p_\nu(b)$, where $p_\nu:A\to (0,1)$ is
a positive probability vector. Finally, let $M_\nu:A\times A\to (0,1)$ be the
probability transition matrix given by
\begin{equation}\label{eq:exotransition}
M_\nu(a',a)=\sum_{b:\,\pi_a(b)=a'}p_\nu(b)\equiv p_\nu(\pi^{-1}_aa'),
\end{equation}
which is the probability of obtaining a word starting
with $a'$ by substitution of the letter $a$, which coincides in 
this case with the probability of obtaining the word $a'a$. This matrix corresponds 
to $M_{1,1}$ in the proof of Theorem~\ref{teo:decayofcorrelation}, therefore
the one-marginal $\mu_1$ of any ${\mathbb S}_\nu$-invariant state $\mu$ is 
given by the unique probability vector $q_\mu:A\to (0,1)$ satisfying
$M_\nu q_\mu=q_\nu$.

\medskip\noindent
\begin{proposition}[{\bf Primitivity}]
For $S$ and $\nu$ as before, the random substitution
${\mathbb S}_\nu$ is primitive.
\end{proposition}

\begin{proof}
It is enough to prove that, for each $\ell\in{\mathbb N}$ there exists
$n_\ell\in {\mathbb N}$ such that for each $a,b\in 2^\ell$ and $N\geq n_\ell$,
there exists a sequence of substitutions 
$s^{(0)},s^{(1)},\ldots,s^{(n)}\in {S^{2^\ell}}$ such that
\begin{equation}\label{eq:requirement}
s^{(n)}\circ s^{(n-1)}\circ\cdots\circ s^{(1)}(b)\sqsupseteq a.
\end{equation}
For this notice that for each $m\in{\mathbb N}$ and $c\in A^{2^m}$,
\[
c=\sigma_{d_1}\sigma_{d_3}\cdots \sigma_{d_{2^m}-1}(c_2c_4\cdots c_{2^m}),
\]
where $\pi_{c_{2k}}(d_{2k-1})=c_{2k-1}$ for each $1\leq k\leq 2^{m-1}$. Hence,
by induction on $m$, we deduce the existence of substitutions
$\hat{s}^{(0)}\in S, \hat{s}^{(1)}\in S^2, \ldots, 
\hat{s}^{(\ell-1)}\in {S^{2^{\ell-1}}}$ such that
\[
a=
\hat{s}^{(\ell-1)}\circ \hat{s}^{(\ell-2)}\circ\cdots\circ \hat{s}^{(0)}(a_{2^\ell}).
\]
On the other hand, since for $c,c'\in A$ we have 
$c=(\sigma_{d}(c'))_1$, with $d=\pi_{c'}^{-1}(c)$, we deduce that for each $b_1\in A$
and $m\geq 1$ there exist a sequence
$\bar{s}^{0},\bar{s}^{(1)},\ldots,\bar{s}^{(m-1)} \in S$ such that 
\[a_{2^{\ell}}=
\left(
\bar{s}^{(m-1)}\circ\bar{s}^{(m-1)}\circ\cdots\circ \bar{s}^{(0)})(b_1))\right)_1.
\]
Finally, for each $n=\ell+m$, any sequences  
$s^{(0)},s^{(1)},\ldots,s^{(n)}\in {S^{2^\ell}}$ such that 
$s^{(k)}\sqsupseteq \bar{s}^{(k)}$ for $0\leq k\leq m-1$ and
$s^{(m+k)}\sqsupseteq \hat{s}^{(k)}$ for $0\leq k\leq \ell-1$, 
satisfies~\eqref{eq:requirement}. The proposition follows with $n_\ell=\ell+1$.
\end{proof}

\medskip\noindent
\subsection{}
For $S$ and $\nu$ as before, the ${\mathbb S}_\nu$-invariant state
can be computed explicitly. We have the following.

\medskip\noindent
\begin{proposition}[{\bf Invariant state}]\label{prop:exoinvariant}
For $S$ and $\nu$ as above, let $M_\nu:A\times A\to (0,1)$ be the 
one-marginal probability transition matrix defined in~\eqref{eq:exotransition} 
and let $q_\nu:A\to (0,1)$ its unique invariant probability vector. 
Then the unique ${\mathbb S}_\nu$-invariant state $\mu_\nu$ is such that
\[
\mu_\nu[a]=q_\nu\left(a_{2^\ell}\right)
\exp\left(\sum_{m=1}^\ell\sum_{k=1}^{2^{\ell-m}}
                \log p_\nu\left(\pi^{-1}_{a_{2^mk}} a_{2^{m-1}(2k-1)}\right) \right)
\]
for each $\ell\in {\mathbb N}$ and $a\in A^{2^\ell}$. Furthermore, 
$\mu_\nu$ is a Gibbs measure for the two-body interaction potential 
$\Phi=\{\Phi_{n,\Lambda}\}_{n\in{\mathbb N},\Lambda\in{\mathcal L}}$, with
\[
\Phi_{2^{m-1}(2k-1),\Lambda}(a)=\left\{\begin{array}{cl} 
           -\log p_\nu\left(\pi^{-1}_{a_{2^mk}} a_{2^{m-1}(2k-1)}\right)  & 
                        \text{ if } \Lambda=\{a_{2^{m-1}(2k-1)},a_{2^mk}\},\\
            0 & \text{ otherwise.}\end{array}\right.
\]
\end{proposition}

\medskip\noindent
\begin{remark}
Notice that the two-body interaction of the previous statement does not
coincide with the one constructed in the proof of Theorem~\ref{teo:gibbsianness}, 
which in particular is positive for sets of arbitrarily large cardinality. Indeed, 
following the aforementioned construction, a straightforward computation gives 
the potential
\[
\Phi_{n,B(n,\ell)}(a)=\frac{1}{2^\ell}\left(\log q_\nu\left(a_{2^{\ell-1}(2q+1)}\right)
-\log p_\nu\left(\pi^{-1}_{a_{2^{\ell}(q+1)}}a_{2^{\ell}(q+1)}\right)\right),
\]
for $2^{\ell}q\leq n <2^{\ell}(q+1)$. It can be verified that 
both potentials determine the same Gibbs measure. Notice as well how the 
two-body interaction organizes in a hierarchical structure similar to a 
binary tree. Indeed, consider
the partition ${\mathbb N}=\sqcup_{m=0}^\infty L_m$, into level sets 
$L_m:=\{2^m(2k-1)\}_{k\in {\mathbb N}}$. Then, each site $2^m(2k-1)$ in the 
$m$-th level set interacts with the site $2^{m+1}k$ situated in an upper level, 
depending on binary decomposition of $k$, and with $2^{m-1}(4k-3)$ in the $(m-1)$-th
level. These kind of interactions were introduced by Dyson in~\cite{1969Dyson},
and since then they have been widely studied in the context of statistical mechanics.
\end{remark}

\medskip\noindent
\begin{proof}
The first claim follows by induction on $\ell$. For $\ell=0$ the claim
is clearly true since $\mu_\nu[a]_1=q_\nu(a)$ for each $a\in A$. Assuming that 
the claim holds for $\ell\geq 1$, and taking into account that 
$\mu_\nu[a]=\mu_\nu[a_2a_4\cdots a_{2^\ell}]\prod_{k=1}^{2^{\ell-1}}
p_\nu(\pi^{-1}_{a_{2k}}a_{2k-1})$ for all $a\in A^{2^\ell}$, then
\begin{align*}
\mu_\nu[a]&=q_\nu\left(a_{2^{\ell}}\right)
\exp\left(\sum_{m=1}^{\ell-1}\sum_{k=1}^{2^{\ell-1-m}}
                \log p_\nu\left(\pi^{-1}_{a_{2^{m+1}k}} a_{2^{m}(2k-1)}\right) \right)
                \exp\left(\sum_{k=1}^{2^{\ell-1}}  
        \log p_\nu\left(\pi^{-1}_{a_{2k}},a_{2k-1}\right) \right)  \\
        &=q_\nu\left(a_{2^{\ell}}\right)
\exp\left(\sum_{m=2}^{\ell}\sum_{k=1}^{2^{\ell-m}}
                \log p_\nu\left(\pi^{-1}_{a_{2^m k}} a_{2^{m-1}(2k-1)}\right) \right)
                \exp\left(\sum_{k=1}^{2^{\ell-1}}  
        \log p_\nu\left(\pi^{-1}_{a_{2k}} a_{2k-1}\right) \right)  \\
        &=q_\nu\left(a_{2^\ell}\right)
\exp\left(\sum_{m=1}^\ell\sum_{k=1}^{2^{\ell-m}}
                \log p_\nu\left(\pi{-1}_{a_{2^m k}}a_{2^{m-1}(2k-1)}\right) \right),
\end{align*}
and the claim follows.

\medskip\noindent
For the second claim it is enough to notice that for $\Lambda\in{\mathcal L}$,
and $N\geq 2\max\Lambda$, the value of 
$\mu_\nu\left(\left[a_\Lambda\right]|
\left[a_{\{1,\ldots,N\}\setminus\Lambda}\right]\right)$ 
does not depend of $N$. Indeed, this value depends only on terms involving couples 
$\{2^{m-1}(2k-1),2^m k\}$ intersecting $\Lambda$. 
A direct computation leads to
\begin{align*}
\lim_{N\to\infty}
\mu_\nu\left(\left[a_\Lambda\right]|
\left[a_{\{1,\ldots,N\}\setminus\Lambda}\right]\right)=
\frac{\exp\left(\sum_{2^{m-1}(2k-1),2^m k\}\cap\Lambda\neq\emptyset}
                \log p_\nu\left(\pi^{-1}_{a_{2^m k}}a_{2^{m-1}(2k-1)}\right)\right) }{
\sum_{c_\Lambda\in A^\Lambda}
\exp\left(\sum_{2^{m-1}(2k-1),2^m k\}\cap\Lambda\neq\emptyset}
   \log p_\nu\left(\pi^{-1}_{\hat{c}_{2^m k}}\hat{c}_{2^{m-1}(2k-1)}\right)\right)
     },
\end{align*}
where $\hat{c}_n=c_n$ if $n\in \Lambda$ and $\hat{c}_n=a_n$ otherwise,
and the claim follows.
\end{proof}

\medskip\noindent
\subsection{}
We can supply an explicit estimation of the decay of correlations
for the two-body interactions we are analyzing here. 
For $\nu$ as before, let $p_{\min}:=\min_{a\in A}p_\nu(a)$, 
$p_{\max}=\max_{a\in A}p_\nu(a)$, and $\Delta p_{\nu}:=p_{\max}-p_{\min}$. 
We have the following. 

\medskip\noindent
\begin{proposition}[{\bf Decay of correlations}]
For $S$ and $\nu$ be as before, there exists $C > 0$ such that
\[
1-C\, n^{-|\log_2\Delta p_\nu|}\leq 
\frac{\mu_\nu([a]_1\cap [b]_n)}{\mu_\nu[a]_1\mu_\nu[b]}
\leq 1+C\, n^{-|\log_2\Delta p_\nu|}
\]
for all $n\in {\mathbb N}$ and each $a,b\in A$.
\end{proposition}

\medskip
\begin{proof}
For each $n\in{\mathbb N}$, let $\ell(n)=\min\{\ell\in{\mathbb N}:\ 2^\ell\geq n\}$,
and write $n=\sum_{i=0}^{\ell(n)}\epsilon_k 2^k$ using the binary expansion with 
$\ell(n)+1$ digits. Define $k(n)=\sum_{i=0}^{\ell(n)}(1-\epsilon_i)\, \mod\ell(n)$,
which is zero in the case $n=2^{\ell(n)}$ or the number of zeros in the binary expansion
of $n$ with $\ell(n)+1$. It is easy to check that $k(n)$ is the number of interactions 
increasing in level, needed to reach site $2^{\ell(n)}$ starting from site $n$. 
It is even easier to verify that $\ell(n)$ is the number of interactions in a
path leading from site 1 to site $2^{\ell(n)}$. 
Using the explicit expression for
$\mu_\nu[a]$ obtained in Proposition~\ref{prop:exoinvariant}, and adding up all 
the interactions $M_\nu(a_{2^{m-1}(2k-1)},a_{2^mk})$ disconnected
from sites $1$ or $n$, we obtain 
\begin{equation}\label{eq:partialquot}
\frac{\mu_\nu([a]_1\cap [b]_n)}{\mu_\nu[a]_1\mu_\nu[b]}=
\frac{\sum_{a'\in A} M^{\ell(n)}_\nu(a,a')M^{k(n)}_\nu(b,a')q_\nu(a')}{q_\nu(a)q_\nu(b)},
\end{equation}
for each $a,b\in A$.

\medskip\noindent
Now, according to Birkhoff's version of Perron-Frobenius Theorem, and taking into account 
that $M_\nu >0$, we have
\[
M^{\ell(n)}_\nu(a,a')=
\left(M^{\ell(n)-1}X_{a'}\right)(a)
\lessgtr q_\nu(a)\exp\left(\pm\frac{d\left(X_{a'},M_\nu X_{a'}\right)
                                  }{\tau_\nu-\tau_\nu^2}\tau_\nu^{\ell(n)}\right),
\]
where $X_{a'}: A\to (0,1)$ is the probability vector such that $X_{a'}(c)=
M_\nu(a',c):= p_\nu(\pi^{-1}_{c}(a'))$,
$\tau_\nu\in [0,1)$ is Birkhoff's contraction coefficient of $M_\nu$ and $d(X,Y)$ is
Hilbert's projective distance between probability vectors $X,Y:A\to (0,1)$, 
\[
\hat{\rho}(X,Y)=\log\left(\max_{c,c'\in A}\frac{X(c)Y(c')}{X(c')Y(c)}\right).
\]
In our case we have explicit bounds for all the terms involved. Indeed, 
$\tau_\nu:=(1-\delta_\nu)/(1+\delta_\nu)$ with
\[
\delta_\nu :=\min_{a,b,c,d\in A}\sqrt{\frac{M_\nu(a,b)M_\nu(c,d)}{M_\nu(a,d) M_\nu(c,b)}}
           \leq \frac{p_{\min}}{p_{\max}}, 
\]
therefore $\tau_\nu\leq \Delta p_\nu= p_{\max}-p_{\min}$ and 
$\Delta p_\nu-\Delta p_\nu^2 \geq \tau_\nu-\tau_\nu^2$. On the other hand,
\[
\hat{\rho}\left(X_{a'},M_\nu X_{a'}\right)=
\log\left(\max_{c,c'\in A}
\frac{p_\nu(\pi_c^{-1}(a')
     \sum_{\hat{c}\in A} p_\nu(\pi_a^{-1}(\hat{c}))p_\nu(\pi_{\hat{c}}^{-1}(c'))
                        }{p_\nu(\pi_{c'}^{-1}(a'))
    \sum_{\hat{c}\in A} p_\nu(\pi_a^{-1}(\hat{c}))p_\nu(\pi_{\hat{c}}^{-1}(c))}\right)
\leq 2\log\left(\frac{p_{\max}}{p_{\min}}\right).                                  
\]

\medskip\noindent From the above computations it follows that 
\[
M^{\ell(n)}_\nu(a,a')
\lessgtr q_\nu(a)\exp\left(\pm 2\frac{\log(p_{\max}/p_{\min})
               }{\Delta p_\nu-(\Delta p_\nu)^2}(\Delta p_\nu)^{\ell(n)}\right).
\]
By taking $C > 0$ sufficiently large, we obtain
\[
M^{\ell(n)}_\nu(a,a')
\lessgtr q_\nu(a)\left(1\pm C_0\,\Delta p_\nu^{\ell(n)}\right). 
\]
Finally, since $\sum_{a'\in A}M^{k(n)}_\nu(b,a')q_\nu(a')=q_\nu(b)$ for $k(n)$ arbitrary,
and $\ell(n)\geq \log_2 n$, if follows that
\[
\frac{\mu_\nu([a]_1\cap [b]_n)}{\mu_\nu[a]_1\mu_\nu[b]}=
\frac{\sum_{a'\in A} M^{\ell(n)}_\nu(a,a')M^{k(n)}_\nu(b,a')q_\nu(a')}{q_\nu(a)q_\nu(b)}
\lessgtr 1\pm  C\, n^{-|\log_2 \Delta p_\nu|}.
\]
\end{proof}

\medskip\noindent
\begin{remark}
The constant $C$ in the previous proposition can be explicitly bounded 
as follows. By convexity, and taking into account that $(\Delta p_\nu)^{\ell(n)}\in (0,1)$ 
for all $n\in {\mathbb N}$, then  
\begin{align*}
\exp\left(2\frac{\log(p_{\max}/p_{\min})
                }{\Delta p_\nu-(\Delta p_\nu)^2}\Delta p_\nu^{\ell(n)}
     \right)
&\leq 1+\left(\exp\left(2\frac{\log(p_{\max}/p_{\min})}{\Delta p_\nu-(\Delta p_\nu)^2}
                   \right)-1\right)(\Delta p_\nu)^{\ell(n)}\\
\exp\left(-2\frac{\log(p_{\max}/p_{\min})
                 }{\Delta p_\nu-(\Delta p_\nu)^2}\Delta p_\nu^{\ell(n)}\right)
&\geq 1-2\frac{\log(p_{\max}/p_{\min})
             }{\Delta p_\nu-(\Delta p_\nu)^2}(\Delta p_\nu)^{\ell(n)}\\
&\geq 1-\left(\exp\left(2\frac{\log(p_{\max}/p_{\min})
           }{\Delta p_\nu-(\Delta p_\nu)^2}\right)-1\right)(\Delta p_\nu)^{\ell(n)}.
\end{align*}
Hence, it is enough to take 
\[C=\left(\exp\left(2\frac{\log(p_{\max}/p_{\min})
           }{\Delta p_\nu-(\Delta p_\nu)^2}\right)-1\right)=
    \left(
    \left(\frac{p_{\max}}{p_{\min}}\right)^{\frac{2}{\Delta p_\nu-(\Delta p_\nu)^2}}
 -1\right).    
\]
\end{remark}

\medskip\noindent 
\begin{remark}[Ising-like interaction]
Although the polynomial law obtained in Theorem~\ref{teo:decayofcorrelation} is only an
upper bound for the decay of correlations, there is an example where it gives the exact decay rate. 
For this consider the particular case of a two-body interaction in $A={\mathbb Z}_2$ 
corresponding to random substitution defined by
$S=\{\sigma_a(b)= b+a \mod 2:\,  a,b\in{\mathbb Z}_2\}$ and the Bernoulli measure 
$\nu[a]_n=p_\nu(a)$ for each $a\in {\mathbb Z}_2$ and $n\in{\mathbb N}$. 
The ${\mathbb S}_\nu$-invariant state is the unique Gibbs measure for the two-body potential
\[
\Phi_{n,\{n,n'\}}(a)=\left\{\begin{array}{cl}
       -\log(p) & \text{ if } \{n,n'\}=\{2^{m-1}(2k-1),2^mk\} \text{ and } a_n=a_{n'},\\
      -\log(1-p) & \text{ if } \{n,n'\}=\{2^{m-1}(2k-1),2^mk\} \text{ and } a_n\neq a_{n'},\\
                     0 & \text{ otherwise.}
\end{array}\right.
\]
For these substitutions, the resulting one-marginal transition matrix $M_\nu$ 
is double-stochastic, therefore $q_\nu=[1/2, 1/2]^\dag$. It is also symmetric with
spectrum $\lambda_0=1 > \lambda_1=2p-1$. 
An easy computation gives
\[
M_\nu^{\ell(n)}=\frac{1}{2}\left(
\begin{matrix} 1+(2p-1)^{\ell(n)}&1-(2p-1)^{\ell(n)}\\1-(2p-1)^{\ell(n)}&1+(2p-1)^{\ell(n)} 
\end{matrix}\right).   
\]
Finally, using Equation~\eqref{eq:partialquot},
which holds for the general two-body interaction, it follows that
\[
\left|\frac{\mu_\nu([a]_1\cap [b]_n)}{\mu_\nu[a]_1\mu_\nu[b]}-1\right|=
|2p-1|^{\ell(n)}\in
\left[\Delta p_\nu n^{-|\log_2\Delta p_\nu|},n^{-|\log_2\Delta p_\nu|}\right].
\] 
\end{remark}


\bigskip
\section{Final comments}
\subsection{}
In Theorem~\ref{teo:projectiveapproximation}, the speed of projective convergence 
depends on the Birkhoff's coefficient and the primitivity index of the transition 
matrices $M_N$. This convergence could be, in principle, very slow. For the random
substitutions studied in Section~\ref{sec:two-bodyinteractions}, the primitivity 
index $m_N$ of the matrix $M_N$ is of the order of $\log N$, while its Birkhoff's
contraction coefficient is of the order of $1-(p_{\min})^{N}$, with this we obtain
a convergence of the order of $(\log \ell)^{-1}$. Indeed, a more precise computation 
leads to
\[
\rho(\mu^{(\ell)},\mu_\nu)\leq 
\frac{(\epsilon + \log(p_{\max}/p_{\min}))\max_{a,b\in A}|\log(p_\nu(a)/q_\nu(b))|}
     {\log(\ell)},
\]
for every $\epsilon >0$ and all $\ell$ sufficiently large.

\medskip\noindent
\subsection{}
As mentioned before, the polynomial bound for the decay of correlations holds for 
non-constant length substitutions. The proof of Theorem~\ref{teo:decayofcorrelation}
can be adapted to the case $1<\ell_S<L_S$. In this case $\mu_\nu([a]_1\cap [b]_n)$ is
determined by $\mu_\nu[c]_1$ after $\log n/\log \ell_S+1$ iterations of the random
substitution. In the chain of substitutions, the paths connecting site $1$ and site
$n$ become independent after $\log n/\log \ell_S-\log n/\log L_S+1$ iterations, and from
this we obtain a bound
\[ 
\left|\frac{\mu_\nu([a]_1\cap [b]_n)}{\mu_\nu[a]_1\,\mu_\nu[b]_n}
                -1\right|={\mathcal O}\left(n^{-\gamma}\right).
\]
with $\gamma$ and $C$ as in the referred theorem.

\medskip\noindent On the other hand, the ${\mathbb S}_\nu$-invariant
state does not appear to be a Gibbs measure for general non-constant length substitutions. 
In this case, the iterates $\mu^{\ell}:={\mathbb S}_\nu^{\circ \ell}$ are not Markovian,
and the projective convergence cannot be ensure. It would be interesting to exhibit a
concrete example where non-Gibbsianness can be established.

\medskip\noindent
\subsection{}
Our setting and several of the outcoming results can be adapted to substitutions on 
infinite graphs other that ${\mathbb N}$. We can consider, for instance, 
constant-volume substitutions in ${\mathbb N}^d$, replacing a letter by a rectangular 
array, and carry on, {\it mutatis mutandis}, all the preceding computations. 
For other infinite graphs or for variable-volume substitutions, further considerations
have to be taken into account.

\section{Acknowledgements}
\noindent We thank CONACyT-M\'exico and Fundaci\'on Marcos Moshinsky 
for their financial support 
through grant CB-2014-237324-F and through ``C\'atedra Marcos Moshinsky 2016'' 
respectively. 
The final stage of the work was done during a visit of E.U. on CPhT-\'Ecole 
Polytecnique, during which he benefited form the financial support of 
\'Ecole Polytechnique, 
and the hospitality and scientific advise of Prof. J.-R. Chazottes. 
C.M. thanks the Instituto de F\'\i sica UASLP for the warm hospitality during 
a one-month visit at the early stage of this work which was partially 
supported by the CONICYT-FONDECYT Postdoctoral Grant No. 3140572.
\newpage
\appendix

\section{Primitivity}\label{primitivity}

\medskip\noindent Let $S:=\{\sigma:A\to A+\}$ be a collection of substitutions. 
For ${\mathcal L} \subset A^+$ and $n\in {\mathbb N}$, let
\[
S^{\circ n}({\mathcal L}):=\left\{b=s^{(n)}\circ s^{(n-1)}\circ\cdots\circ s^{(1)}(a):
\ a\in {\mathcal L},\, s^{(1)},s^{(2)},\ldots,s^{(n)}\in S^+\right\}.
\] 
For ${\mathcal L},{\mathcal L}'\subset A^+$, we will say that 
${\mathcal L}'\sqsupseteq {\mathcal L}$ whenever for each $a'\in {\mathcal L}'$ 
there exists $a\in {\mathcal L}$ such that $a'\sqsupseteq a$. Using this notation
we can reformulate our definition of primitivity. We clearly have that
$S$ is primitive if for each $N$ there exists $n_N$ such that 
\[
S^{\circ n}(\{c\})\sqsupseteq A^N,
\]
for each $c\in A^N$.

\medskip\noindent For a collection $S:=\{\sigma:A\to A^+\}$ of substitutions,
{\bf\em the one-symbol transition matrix} $M_S\in M_{A\times A}(\{0,1\})$ is given by
\[
M_S(a,b)=\left\{\begin{array}{ll} 1 &\text{ if } S(a)\sqsupseteq\{b\},\\   
                                    0 &\text{ otherwise.}\end{array}\right.
\] 

\medskip\noindent 
We will say that $a\in A$ is a {\bf\em sliding symbol}, with respect to $S$, if 
$a^{p+1}\in S(\{a\})$ for some $p\geq 1$, $S(\{a\})\sqsupseteq  A$,
$S(A)\supset A^q$ for some $q\geq 1$.

\medskip\noindent We have the following.

\medskip\noindent
\begin{proposition} 
For a collection of substitutions $S:=\{\sigma:A\to A^+\}$ 
to be primitive, it is enough that the one-symbol transition 
matrix $M_S$ be primitive and that there exist a sliding symbol $a\in A$. 
\end{proposition}

\begin{proof}
The result is a direct consequence of these two claims. 
\begin{itemize}
\item[(A)] Let $n_0$ be the primitivity index of $M_S$.
For each $N\in {\mathbb N}$, $c_1^N\in A_1^N$ and $n\geq n_0+\log_{1+p}(N)$, 
$S^{\circ n}(\{c_1^N\})\sqsupseteq \{a^N\}$.
\item[(B)] For each $N$ and $n\geq \max(N,n_0)$, 
$S^{\circ n}(\{a^N\})\sqsupseteq A^N$.
\end{itemize}

\medskip \noindent Claim (A) follows from the primitivity of $M_S$  
and the fact that $a^{p+1}\in S(a)$. 

\medskip \noindent Claim (B) can be easily proved 
by induction. Indeed, by hypothesis $S(\{a\})\sqsupseteq A$, which ensures the 
validity of the claim form $N=1$. 
Assuming $S^{\circ n}(\{a^N\})\sqsupseteq A^N$
for all $n\geq \max(N,n_0)$, if $q>1$, then necessarily 
$S^{\circ (n+1)}(\{a^{N+1}\})\sqsupseteq S(A^N)\supset A^{qN}\sqsupseteq A^{N+1}$.
If on contrary $q=1$, then
\[
S^{\circ (n+n_0)}(\{a^{N+1}\})\sqsupseteq S^{\circ n}(A^N)S^{\circ n_0}(\{c\})
=A^NS^{\circ n_0}(\{c\})
\]
for some $c\in A$. Finally, by primitivity of $M_S$, we have 
$S^{\circ n_0}(\{c\})\sqsupseteq A$, and the claim (B) follows.

\medskip\noindent 
The proposition follows from claims (A) and (B), with primitivity index 
$n_N=N+2n_0+\log_{1+p}(N)$.

\end{proof}

\medskip\noindent\begin{remark} 
Let us remark that while the condition ``$M_S$ is primitive'' is 
necessary for the primitivity of $S$, the other condition, 
``there exist a sliding symbol'', does not seem to be necessary. 
The conditions on the previous proposition are satisfied, for instance
for the substitutions leading to two-body interaction potentials considered in
Section~\ref{sec:two-bodyinteractions} and the non-constant length substitution 
studied in~\cite{2013SalgadoUgalde}
\end{remark}

\section{Boundedness}\label{boundedness}

\medskip\noindent Let $S:=\{\sigma:A\to A+\}$ be a primitive constant-length 
substitution of length $L > 1$. For each $a\in A$ and $1\leq j\leq L$, let
$S(a)_j:=\{\sigma(a)_j\in A:\ \sigma\in S\}$.
We will say that ${\mathcal S}$ has {\bf \em bundle structure} if for each 
$a\in A$ we have $S(a)=\prod_{j=1}^LS(a)_j$. 

\medskip\noindent Let $\nu\in{\mathcal M}(A^{\mathbb N})$ 
be a product measure, and for each $n\in{\mathbb N}$, let $\nu_n$ 
denote the one-marginal at position $n$. We 
will say that $\nu$ has {\bf\em bounded dispersion} if
$\sup_{n\in{\mathbb N}}\rho(\nu_1,\nu_n) < \infty$.

\medskip\noindent 
\begin{remark}
Notice that the random substitutions leading to a two-body interaction
potential have bundle structure, and obviously any Bernoulli measure has bounded
dispersion.
\end{remark}

\medskip\noindent
\begin{proposition} 
Let $S$ be a primitive constant-length substitution of length $L > 1$ and 
$\nu\in {\mathcal M}(S^{\mathbb N})$ a fully supported product measure. 
Let $\mu_\nu$ the unique ${\mathbb S}_\nu$-invariant state and 
$\mu$ the unique product measure such that $\mu[a]_n=\mu_\nu[a]_n$ 
for each $n\in {\mathbb N}$ and $a\in A$. If $S$ has bundle structure and
$\nu$ has bounded dispersion then $\rho(\mu,\mu_\nu)<\infty$.
\end{proposition}

\begin{proof}
Let us start by fixing some notations. 
First, let $q_\nu\in (0,1)^A$ 
denote the probability vector corresponding to one-marginal of $\mu_\nu$
at position 1. For each $\ell\in {\mathbb N}$, 
$M^{(\ell)}:A^{L^{\ell}}\times A^{L^{\ell+1}}\to [0,1]$
is such that
\[M^{(\ell)}(a,b):=\sum_{s(b)=a} \nu[s]=\prod_{n=1}^{L^\ell}
\nu_n\left\{\sigma\in S:\ \sigma(b_n)=a_{(n-1)L+1}^{nL} \right\}.\]
Clearly $\sum_{b\in A^{L^\ell}}\mu_\nu[a]=M^{(\ell)}(a,b)\mu_\nu[b]$ 
for each $a\in A^{L^\ell+1}$ 
and $b\in A^{L^\ell}$.

\medskip \noindent For each $n\in{\mathbb N}$ and $1\leq j\leq L$, 
let $M_{n,j}:A\times A\to [0,1]$ be given by
$M_{n,j}(a,c)=\sum_{\sigma(c)_j=a}\nu[\sigma]_n$. 
With this define 
$\bar{M}^{(\ell)}:A^{L^{\ell}}\times A^{L^{\ell+1}}\to [0,1]$ by
\[
\bar{M}^{(\ell)}(a,b)
      :=\prod_{n=1}^{L^\ell}\prod_{j=1}^L M_{n,j}(a_{(n-1)L+j},b_n)
      =\prod_{n=1}^{L^\ell} \left(\prod_{j=1}^L
      \nu_n\left\{\sigma\in S: \ \sigma(b_n)_j=a_{(n-1)L+j}\right\}
                         \right).
\]
It is not difficult to verify that 
$\mu[a]=\sum_{b\in A^{L^\ell}}\bar{M}_1^{(\ell)}(a,b)\mu[b]$
for each $a\in A^{L^{\ell+1}}$ and $b\in A^{L^\ell}$. 

\medskip\noindent Since $S$ has bundle structure, then
$M_k^{(\ell)}(a,b)\neq 0$ if and only if $\bar{M}_k^{(\ell)}(a,b)\neq 0$. 
The proof of the proposition is base on the following claim: 
\begin{itemize}
\item[] There exists $C>0$ such that 
$\bar{M}^{(\ell)}(a,b) \lessgtr e^{\pm L^\ell C}\ M^{(\ell)}(a,b)$
for each $\ell\in {\mathbb N}$, 
$a\in A^{L^{\ell+1}}$ and $b\in A^{L^\ell}$.
\end{itemize}
The claim is a direct consequence of bounded dispersion. For this, it is enough to
observe that for each $b\in A$ and $a\in A^L$,
\begin{align*}
\frac{\prod_{j=1}^L
      \nu_n\left\{\sigma\in S: \ \sigma(b)_j=a_j\right\}
 }{\nu_n\left\{\sigma\in S:\ \sigma(b)=a \right\}}
&\lessgtr
 \frac{\prod_{j=1}^L
      \nu_1\left\{\sigma\in S: \ \sigma(b)_j=a_j\right\}
 }{\nu_1\left\{\sigma\in S:\ \sigma(b)=a \right\}}e^{\pm (L+1)\rho(\nu_1,\nu_n)}\\
&\lessgtr
 \frac{\prod_{j=1}^L
      \nu_1\left\{\sigma\in S: \ \sigma(b)_j=a_j\right\}
 }{\nu_1\left\{\sigma\in S:\ \sigma(b)=a \right\}}e^{\pm (L+1)\sup_n\rho(\nu_1,\nu_n)}  
\end{align*}
Hence, for each $\ell\in {\mathbb N}$,
$L^{\ell-1}\leq N <L^{\ell}$ and $a\in A^{N}$, we have
\begin{align*}
\frac{\mu_\nu[a]}{\mu[a]}
=\frac{\sum_{b\in A^{L^\ell}: b\sqsupseteq a}\mu_\nu[b]
     }{\sum_{b\in A^{L^\ell}: b\sqsupseteq a}\mu[b]}&=
\frac{
\sum_{b\in A^{L^\ell}: b\sqsupseteq a}
\left(\left(\prod_{k=1}^{\ell-1} M^{(k)}\right)q_\nu\right)(b)
    }{
\sum_{b\in A^{L^\ell}:  b\sqsupseteq a}
\left(\left(\prod_{k=1}^{\ell-1}\bar{M}^{(k)}\right)q_\nu\right)(b)}\\
 &\lessgtr 
 e^{\pm (L+1)\sup_n\rho(\nu_1,\nu_n)\sum_{k=1}^{\ell-1}L^k}
 =e^{\pm (L+1)\,L^\ell\sup_n\rho(\nu_1,\nu_n)},
\end{align*}
therefore $\rho(\mu,\mu_\nu)\leq L(L+1)\sup_n\rho(\nu_1,\nu_n) <\infty$.

\end{proof}

\bigskip
\bibliographystyle{plain}

\end{document}